\documentclass[english,12pt]{article}
\usepackage[latin1]{inputenc}
\usepackage[dvips]{graphics}
\usepackage{amsmath,amsthm,amssymb,babel}
\usepackage{color}
\usepackage{graphicx}
\usepackage{ulem}

\newtheorem{problem}{Problem}

\newcommand{\R}{{\if mm {\rm I}\mkern -3mu{\rm R}\else \leavevmode
\hbox{I}\kern -.17em\hbox{R} \fi}}

\newcommand{\bu}{\mbox{\boldmath{$u$}}}

\newcommand{\bv}{\mbox{\boldmath{$v$}}}

\newcommand{\bx}{\mbox{\boldmath{$x$}}}

\newcommand{\fb}{\mbox{\boldmath{$f$}}}

\newcommand{\bsigma}{\mbox{\boldmath{$\sigma$}}}
\newcommand{\btau}{\mbox{\boldmath{$\tau$}}}

\newcommand{\bvarepsilon}{\mbox{\boldmath{$\varepsilon$}}}
\newcommand{\bnu}{\mbox{\boldmath{$\nu$}}}

\newcommand{\bzero}{\mbox{\boldmath{$0$}}}

\newcommand{\cS}{\mbox{{${\cal S}$}}}
\newcommand{\cR}{\mbox{{${\cal R}$}}}

\def\real{\mathbb{R}}

\newtheorem{theorem}{Theorem}[section]
\newtheorem{lemma}[theorem]{Lemma}
\newtheorem{corollary}[theorem]{Corollary}
\newtheorem{definition}[theorem]{Definition}

\newtheorem{example}[theorem]{Example}

\numberwithin{equation}{section}

\author{\it The authors}

\title{\bf Time-dependent Inclusions and Sweeping Processes in Contact Mechanics}
{\author{Samir Adly$^1$ and Mircea Sofonea$^2$\\[5mm]
		{\it\small $^1$ Laboratoire XLIM,  University of Limoges}\\
		{\it \small 123 Avenue Albert Thomas, 87060 Limoges, France}
		\\
		{\it \small samir.adly@unilim.fr}\\[5mm]
	{\it \small $^2$Laboratoire de Math\'ematiques et Physique}\\
		{\it \small
			University of Perpignan Via Domitia}\\{\it\small 52 Avenue Paul Alduy, 66860 Perpignan,
			France}\\
		{\it \small sofonea@univ-perp.fr}
	}
\date{}

\begin{document}
\maketitle

\vskip8mm
\begin{abstract}
\noindent We consider a class of time-dependent inclusions  in Hilbert spaces for which we state and prove an existence and uniqueness result.
The proof is based on arguments of variational inequalities, convex analysis and fixed point theory.
Then we use this result to prove the unique weak solvability of a new class of Moreau's sweeping processes with constraints in velocity. Our results are useful in the study of mathematical models which describe the quasistatic evolution of  deformable bodies in contact with an obstacle.
To provide some examples we consider three viscoelastic  contact problems which lead to time-dependent  inclusions and sweeping processes in which the unknowns are the displacement and the velocity fields, respectively. Then we apply our abstract results in order to prove the unique weak solvability of the corresponding contact problems.

\end{abstract}

\bigskip \noindent {\bf AMS Subject Classification\,:}   49J40,  47J20, 47J22, 34G25, 58E35, 74M10, 74M15, 74G25.

\bigskip \noindent {\bf Key words\,:}
nonlinear inclusion, sweeping process, contact problem, unilateral constraint, weak solution.
\vskip15mm
\section{Introduction}\label{int}

Contact phenomena with deformable bodies arise in a large variety of industrial settings and engineering applications.  Their classical formulation leads
to  challenging nonlinear boundary value problems in which the unknowns are the displacement and the stress field. Most of these problems include unilateral constraints and represent free boundary problems. For this reason, their mathematical analysis is done by using the so-called weak formulation which, usually, is expressed in terms of variational or hemivariational inequalities in which the unknown is the displacement or the velocity field.  
Comprehensive reference in the field are \cite{AC,DL,EJK,HS,KO,Pa1, Pa2,SM}  and, more recently, \cite{SMBOOK}.

An important number of problems arising in Mechanics, Physics and
Engineering Science leads to mathematical models expressed in terms
of nonlinear time-dependent inclusions. For this reason the mathematical
literature dedicated to this field is extensive and the progress
made in the last  decades is impressive. It concerns both
results on the existence, uniqueness, regularity and behavior of
the solution for various classes of  inclusions as well
as results on the numerical approaches to the solution of the
corresponding problems. Variational and hemivariational inequalities represent 
a class of nonlinear inclusions that are associated with the subdifferential in the sense of convex analysis and the Clarke
subdifferential operator, respectively. They have made the object of various books and surveys, see \cite{HMS, MOSBOOK, NP, Pa2, SM, SMBOOK}, for instance.

The notion of ``sweeping process" was introduced by Jean Jacques Moreau in early seventies, in connexion with the study of displacement-tractions problems for elastic-plastic materials, see \cite{M1,M2,M3,M4}. There, the treatment of both theoretical and numerical aspects of sweeping processes have been
developed and their applications in unilateral mechanics  were illustrated. Since the pioneering works of Moreau, several extensions and generalizations have been considered in literature for which various existence and uniqueness results have been provided. References on the field are \cite{AHT,KM} and, more recently \cite{AH}.

The aim of this paper is two folds. The first one to introduce a new class of time-dependent inclusions and sweeping processes and to study their unique solvability. Here, the novelty arises in the special structure of the problems we consider, which are governed by two nonlinear operators, possible history-dependent, and are defined on a time interval which could be either bounded or unbounded. Moreover, one of the operators appears in the set of constraints. The second aim is to illustrate the use of these results in the study of
mathematical models arising in Contact Mechanics. In contrast with the standard variational formulations considered in the literature, the contact models
we consider here lead to time-dependent inclusions and sweeping processes, which represents the second trait of novelty of this paper.

The paper is structured as follows. In Section \ref{s2} we introduce the notation we use and the preliminaries of convex analysis and nonlinear analysis   we need in the rest of the paper. They include an existence and uniqueness result for elliptic variational inequalities and a fixed point result for almost history-dependent operators, amog others.
In Section \ref{s3} we introduce the time-dependent inclusions and prove their unique solvability, Theorem \ref{t1}.
Then, in Section  \ref{s4} we introduce the sweeping processes we are interested in and prove an existence and uniqueness result,  Theorem \ref{t2}.
Finally, in Sections \ref{s5} and \ref{s6} we illustrate the use of our abstract results in the study of three contact models with  viscoelastic materials, both in the frictionless and frictional case. In this way we provide an example of cross  fertilization between models and applications, in one hand, and the nonsmooth analysis, on the other hand.

\section{Preliminaries}\label{s2}
\setcounter{equation}0

Most of the material presented in this section is standard. Therefore, we introduce it without proofs and restrict ourselves to mention that details on the definitions and statements below can be found in the monographs
\cite{DMP1,ET, KZ, MOSBOOK} as well as in the paper \cite{AH}.

\medskip\noindent
{\bf Elements of convex analysis.} Everywhere in this paper $X$ will represent a real Hilbert space with the inner product $(\cdot,\cdot)_X$
and the associated norm   $\|\cdot\|_X$. Moreover, we denote by
$0_X$ the zero element of $X$ and by $2^{X}$ the set of parts of $X$.
	
	Assume that $J:X\to\, ]-\infty,+\infty]$ is a convex lower semicontinuous function such that $J\not\equiv\infty$, i.e., $J$ is proper. The effective domain of $J$ is the set ${\rm Dom}(J)$ defined by
	\begin{equation*}
	{\rm Dom}(J)=\{ \, u \in X\ :  J(u)<+\infty \, \}.
	\end{equation*}
	The subdifferential of $J$ (in the sense of convex analysis) is the multivalued operator $\partial J:X\to 2^{X}$
	defined by
	\begin{equation}\label{sub}
	\partial J(u)
	=
	\{ \, \xi \in X \ :\  J(v)-J(u)\ge (\xi,v-u)_X  
	\quad \forall\, v \in X \, \}.
	\end{equation}
	An element $\xi \in \partial J (u)$ (if any)
	is called a {\rm subgradient} of $J$ in $u$.
	We recall that if $u\notin{\rm Dom}(J)$ then $\partial J(u)=\emptyset$. For the above function $J$, its Legendre-Fenchel conjugate is defined as
	$J^*:X\to ]-\infty,+\infty]$,
	\[J^*(u^*)=\sup_{u\in X}\,\big((u^*,u)_X-J(u)\big).\]
	Moreover, the following equivalence holds.
	\begin{equation}\label{e}
	u^*\in\partial J(u)\quad\Longleftrightarrow\quad u\in \partial J^*(u^*).
	\end{equation}
	
	Let $C\subset X$ be a nonempty closed convex subset. 
	The function ${\rm I}_C$ defined by 
	$$
	{\rm I}_C(x)=
	\begin{cases}
	0 &\text{{\rm if} \ $x \in C$,} \\
	+\infty &\text{{\rm if} \ $x \notin C$}
	\end{cases}
	$$
	is called the indicator function of $C$. 
	Using (\ref{sub})  it follows that the subdifferential of ${\rm I}_C$ is the multivalued operator $\partial {\rm I}_C:X\to 2^{X}$
	defined by
	\begin{equation}\label{su}
	\partial {\rm I}_C(u)
	=\left\{\begin{array}{l}
	\{ \, \xi \in X\ :\  (\xi,v - u)_X \le 0 
	\ \ \ \forall\,v \in C \, \}\quad{\rm if}\ u\in C,\\[2mm]	
	\emptyset	\quad{\rm if}\ u\notin C.
	\end{array} \right.
	\end{equation}
	As usual in the convex analysis, we denote the subdifferential of the function ${\rm I}_C$ by ${\rm N}_C$, i.e.,  $\partial {\rm I}_C={\rm N}_C$.
	For a given $u\in C$,  the set $\partial {\rm I}_C(u)={\rm N}_C(u)\subset X$ represents the set of outward normals of the convex set at the point $u\in C$. Moreover, it is easy to check that
	\begin{eqnarray}
	&&\label{n1} {\rm N}_C(-u)=-{\rm N}_{-C}(u)\qquad\quad\ \forall\, u\ \mbox{such that}\ u\in -C\\ [2mm]
	&&\label{n2} {\rm N}_C(u+v)={\rm N}_{C-v}(u)\qquad\quad\forall\, u,\,v\ \mbox{such that}\ u+v\in C.
	\end{eqnarray}




	


\medskip\noindent
{\bf Variational inequalities.}
We recall that
an operator $A \colon X \to X$ is said to be
strongly monotone
if there exists $m_A>0$ such that
\begin{equation}\label{A1}
(Au-Av,u-v)_X\ge m_A\|u-v\|_X^2\qquad\forall\, u,\,v\in X.
\end{equation}
The operator $A$ is	
Lipschitz continuous if there exists a constant $L_A>0$ such that
\begin{equation}\label{A2}
\| Au - Av \|_X\le L_A \|u-v\|_X \qquad\forall\, u,\,v\in X.
\end{equation}
A function
	$j \colon K\subset X\to\mathbb{R}$ is said to be lower semicontinuous
	(l.s.c.)\ at $u\in K$ if 
	\begin{equation}\label{lsc}
	\liminf_{n\to\infty} j(u_n)\ge j(u)
\end{equation}
for each sequence $\{u_n\}\subset K$ converging to $u$ in $X$. The
function $j$ is lower semicontinuous (l.s.c.)\  if it is lower semicontinuous at every
point $u\in K$.  We now recall a classical result in the study of variational inequalities.

\begin{theorem}\label{t0}
	Let $X$ be a Hilbert space and assume that $K$ is a nonempty closed convex
	subset of $X$, $A:X\to X$ is a strongly
	monotone Lipschitz continuous operator and  $j:K\to\mathbb{R}$ is a convex lower semicontinuous  function.
	Then, for each $f\in X$, there exists a unique 
	solution of the variational inequality
	\begin{equation}\label{BC21}
	u\in K,\quad(Au,v-u)_X+j(v)-j(u)\geq (f,v-u)_X\quad\forall\,
	v\in K.
	\end{equation}
\end{theorem}

Theorem \ref{t0} will be used in Section \ref{s3} to prove the unique
solvability of our nonlinear inclusion. Its proof is based on the Banach fixed point argument and could be found in \cite{SM}, for instance.

\medskip\noindent
{\bf History and almost history-dependent operators.}  
Everywhere below $I$ will denote either a bounded interval of the form  $[0,T]$ with $T>0$, or the unbounded interval $\mathbb{R}_+=[0,+\infty)$. For a normed space $(Y,\|\cdot\|_Y)$ we denote by $C(I;Y)$ the space of continuous
functions defined on $I$ with values in $Y$, that is, 
\[
C(I;Y)=\{\ v\colon I\to Y \mid v{\rm \ is\ continuous\ }\}.
\]
The case $I=[0,T]$ leads to the space 
$C([0,T];Y)$  which is a normed space equipped with the norm
\[ \|v\|_{C([0,T];Y)} =
\max_{t\in [0,T]}\,\|v(t)\|_Y. \]
If $Y$ is a Banach space, then $C([0,T];Y)$  is a Banach space, too. 
The case $I=\mathbb{R}_+$ leads to the space
$C(\mathbb{R}_+;Y)$. If $Y$ is a Banach space then $C(\mathbb{R}_+;Y)$ can be organized in a canonical way as 
a Fr\'echet space, i.e., a complete metric space in which the
corresponding topology is induced by a countable family of
seminorms. For a subset $K \subset Y$ we still use the symbol $C(I;K)$
for the set of continuous functions defined on 
$I$ with values on $K$.


We also denote by $C^1(I;Y)$ the space of continuously
differentiable functions on $I$ with values in $Y$ and,
we note that $v\in C^1(I;Y)$ if and only if $v\in C(I;Y)$
and $\dot v\in C(I;Y)$ where, here and below, $\dot v$ represents the derivative of the function $v$. Moreover, for a subset $K\subset Y$,
we denote by $C^1(I;K)$ the set of continuously
differentiable functions on $I$ with values in $K$.
For a function $v\in C^1(I;Y)$, the equality below
will be used in various places of this manuscript:
\begin{equation*}
v(t)=\int_0^t\dot{v}(s)\,ds+v(0)
\ \ \mbox{for all} \ \ t\in I.
\end{equation*}

Two important classes of operators defined on the space of continuous functions are provided by the following definition.

\begin{definition} \label{d}
Assume that	$(Y,\|\cdot\|_Y)$ and $(Z,\|\cdot\|_Z)$ are normed spaces. An operator $\cS \colon C(I;Y)\to C(I;Z)$ is called:
	
\medskip	
{\rm a) history-dependent (h.d.)},
	if  for any compact set $\mathcal{J}\subset I$, there exists $L_\mathcal{J}^\mathcal{S}>0$ such that
	\begin{eqnarray}
	&&\label{2.1}\|\cS u_1(t)-\cS u_2(t)\|_Z\le
	L_\mathcal{J}^\mathcal{S}\,\displaystyle\int_0^t
	\|u_1(s)-u_2(s)\|_Y\,ds\\[0mm]
	&&\qquad\mbox{\rm for all} \ \ u_1,\,u_2\in
	C(I;Y),\ \ t\in \mathcal{J}.\nonumber
	\end{eqnarray} 
	
{\rm b) almost his\-to\-ry-de\-pen\-dent (a.h.d.)},
if for any compact set $\mathcal{J}\subset I$, there exists $l_\mathcal{J}^\mathcal{S}\in[0,1)$ and $L_\mathcal{J}^\mathcal{S}>0$ such that
\begin{eqnarray}
&&\label{2.2}\|\cS u_1(t)-\cS u_2(t)\|_Z\le
l_\mathcal{J}^\mathcal{S}\,\|u_1(t)-u_2(t)\|_Y\\[0mm]
&&\quad+L_\mathcal{J}^\mathcal{S}\,\displaystyle\int_0^t	\|u_1(s)-u_2(s)\|_Y\,ds\ \
\mbox{\rm for all} \ \ u_1,\,u_2\in
C(I;Y),\ \ t\in \mathcal{J}.\nonumber
\end{eqnarray}

\end{definition}

\medskip
Note that here and below, when no confusion arises, we use the shorthand notation $\cS u(t)$ to
represent the value of the function $\cS u$ at the point $t$,
i.e., $\cS u(t)=(\cS u)(t)$.  It follows from the previous definition that any h.d. operator is an a.h.d. operator.
History-dependent and almost history-dependent operators arise in Contact Mechanics and Nonlinear Analysis. They have important fixed point properties
which are very useful to prove the solvability of va\-rio\-us classes of nonlinear equations and variational inequalities.

\begin{theorem} \label{t00}
	Let $Y$ be a Banach space and 
	let $\Lambda\colon C(I;Y)\to C(I;Y)$ be an almost history-dependent operator.
	Then, $\Lambda$ has a unique fixed point, i.e., there exists a unique element $\eta^*\in C(I;Y)$
	such that $\Lambda\eta^*=\eta^*$.
\end{theorem}

A proof of Theorem \ref{t00} can be found  in \cite[p. 41--45]{SMBOOK}. There, the main properties of history-dependent and almost history-dependent operators are stated and proved, together with various examples and applications.

\medskip\noindent
{\bf Function spaces.}  Let $d\in\{1,2,3\}$ and denote by $\mathbb{S}^d$  the space of second order symmetric
tensors on $\mathbb{R}^d$ or, equivalently, the space of symmetric
matrices of order $d$. The zero element of the spaces $\mathbb{R}^d$ and $\mathbb{S}^d$ will be denoted by $\bzero$. The  inner product and norm on
$\mathbb{R}^d$ and $\mathbb{S}^d$ are defined by
\begin{eqnarray*}
	&&\bu\cdot \bv=u_i v_i\ ,\qquad\ \
	\|\bv\|=(\bv\cdot\bv)^{\frac{1}{2}}\qquad\,
	\forall \,\bu=(u_i),\, \bv=(v_i)\in \mathbb{R}^d,\\[0mm]
	&&\bsigma\cdot \btau=\sigma_{ij}\tau_{ij}\ ,\qquad
	\|\btau\|=(\btau\cdot\btau)^{\frac{1}{2}} \qquad \forall\,
	\bsigma=(\sigma_{ij}),\,\btau=(\tau_{ij})\in\mathbb{S}^d,
\end{eqnarray*}
where the indices $i$, $j$ run between $1$ and $d$ and,
unless stated otherwise, the summation convention over repeated
indices is used.

Consider now a
bounded domain $\Omega\subset\mathbb{R}^d$ with a
Lipschitz continuous boundary $\Gamma$ and let  $\Gamma_1$ be a
measurable part of $\Gamma$ such that ${ meas}\,(\Gamma_1)>0$. 
In Sections \ref{s5} and \ref{s6} of this paper we use the
standard notation for Sobolev and Lebesgue spaces associated to a
 bounded domain $\Omega\subset\mathbb{R}^d$ ($d=1,2,3$), with a
 Lipschitz continuous boundary $\Gamma$. In particular, we use the spaces  $L^2(\Omega)^d$, $L^2(\Gamma_2)^d$,
$L^2(\Gamma_3)$, $L^2(\Gamma_3)^d$  and $H^1(\Omega)^d$, endowed with their canonical inner products and associated norms.
Moreover, for an element $\bv\in H^1(\Omega)^d$ we usually  write $\bv$ for the trace $\gamma\bv\in L^2(\Gamma)^d$ of
$\bv$ to $\Gamma$. In addition, we consider the following
spaces:
\begin{eqnarray*}
	&&V=\{\,\bv\in H^1(\Omega)^d:\  \bv =\bzero\ \ {\rm on\ \ }\Gamma_1\,\},\\
	&&Q=\{\,\bsigma=(\sigma_{ij}):\ \sigma_{ij}=\sigma_{ji} \in L^{2}(\Omega)\,\}.
\end{eqnarray*}

The spaces $V$ and $Q$ are real Hilbert spaces
endowed with the canonical inner products given by
\begin{equation}
(\bu,\bv)_V= \int_{\Omega}
\bvarepsilon(\bu)\cdot\bvarepsilon(\bv)\,dx,\qquad
( \bsigma,\btau )_Q =
\int_{\Omega}{\bsigma\cdot\btau\,dx}.
\end{equation}
Here and below $\bvarepsilon$   
represents the deformation  operator,
that is
\[
\bvarepsilon(\bu)=(\varepsilon_{ij}(\bu)),\quad
\varepsilon_{ij}(\bu)=\frac{1}{2}\,(u_{i,j}+u_{j,i}),\]
the index that follows a comma denoting the
partial derivative with respect to the corresponding component of
the spatial variable $\bx$, i.e.,\ $u_{i,j}={\partial
	u_i}/{\partial x_j}$. 
The associated norms on these spaces  are denoted by
$\|\cdot\|_V$ and $\|\cdot\|_{Q}$,
respectively. Recall that the completeness of the space $V$ follows from the assumption
${ meas}\,(\Gamma_1)>0$ which allows the use of Korn's inequality.
Let  ${\bnu}=(\nu_i)$ be the outward unit normal at $\Gamma$. 
For any element $\bv\in V$,  we denote by $v_\nu$ and $\bv_\tau$ its normal and
tangential components on $\Gamma$ given by
$v_\nu=\bv\cdot\bnu$ and $\bv_\tau=\bv-v_\nu\bnu$, respectively. 
In addition, we recall that the Sobolev
trace theorem yields
\begin{equation}\label{trace}
\|\bv\|_{L^2(\Gamma_3)^d}\le c_0\,\|\bv\|_{V}\quad
{\rm for\ all}\ \bv \in V,
\end{equation}
$c_0$ being a positive constant which depends on $\Omega$, $\Gamma_1$ and $\Gamma_3$.

Next, for a
regular stress function $\bsigma:\Omega\to \mathbb{S}^d$,  the following Green's formula
holds:
\begin{equation}
\int_\Omega\,\bsigma\cdot\bvarepsilon(\bv)\,dx+\int_\Omega\,{\rm
	Div}\,\bsigma\cdot\bv\,dx = \int_\Gamma\bsigma\bnu \cdot\bv\,da
\quad {\rm\ for\ all }\ \bv\in H^1(\Omega)^d. \label{Green}
\end{equation}
Here and below in this paper ${\rm Div}$ denotes the divergence operator, i.e.,
${\rm	Div}\,\bsigma=(\sigma_{ij,j})$.

Finally, we introduce the space of fourth order tensors defined by 
\begin{equation}\label{Qi} {\bf Q_\infty}=\{\, {\cal E}=(e_{ijkl})\mid 
{e}_{ijkl}={e}_{jikl}={e}_{klij} \in L^\infty(\Omega),\ 1\le
i,j,k,l\le d\,\}\,.
\end{equation}
It is a Banach space endowed with the norm 
\begin{equation*}\label{**}
\displaystyle \|{\cal{E}}\|_{\bf Q_{\infty}}=\max_{0\le i,j,k,l\le
	d}\|{e}_ {ijkl}\|_{L^{\infty}(\Omega)}. \end{equation*} \noindent
Moreover it is easy to see that
\begin{equation}\label{pmp}
\|{\cal{E}}\btau\|_{Q}\le d\,\|{\cal{E}}\|_{\bf Q_{\infty}}
\|\btau\|_Q\ \ \mbox{for all} \ \ 
{\cal{E}}\in{\bf Q_{\infty}},\ \btau\in Q.
\end{equation}
This inequality will be repeatedly used in Sections \ref{s5} and \ref{s6} to provide the history-dependent feature of the relaxation tensors.

\section{Time-dependent inclusions}\label{s3}
In this section we state and prove existence and uniqueness results for time-dependent inclusions with nonlinear operators and, in particular, with history-dependent operators. The functional framework is the following: besides the Hilbert space $X$ we consider a real Hilbert space $Y$  endowed with the  inner  product $(\cdot,\cdot)_Y$
and the associated norm   $\|\cdot\|_Y$. We denote by $Y\times X$ the product space of $Y$ and $X$, endowed with the inner product  product $(\cdot,\cdot)_{Y\times X}$
and the associated norm   $\|\cdot\|_{Y\times X}$. Moreover, we assume the following.

\bigskip
\noindent 
$(K)$\qquad	  $K\subset X$ is a nonempty closed convex cone (and, therefore, $0_X\in K$).

\bigskip
\noindent 
$(A)$\qquad	   $\left\{\begin{array}{l} A:X\to X\ \ \mbox{ is a strongly monotone Lipschitz continuous operator,}\\
\mbox{  i.e., it satisfies conditions}\ 
(\ref{A1})\ \mbox{and}\ (\ref{A2})\ \mbox{ with}\  m_A>0\ \mbox{and}\ L_A>0,\\ \mbox{\ respectively}.
\end{array} \right.$

\bigskip
\noindent 
$(\cR)$\qquad\ $\left\{\begin{array}{l}
\cR:C(I;X)\to C(I;Y)\ \ \mbox{and for any compact set}\\[0mm]	
 \mathcal{J}\subset I,\ \mbox{there exists}\ l_\mathcal{J}^\mathcal{R}>0\ \mbox{and}\ L_\mathcal{J}^\mathcal{R}>0\   \mbox{such that}\\[3mm]
\ \|\cR u_1(t)-\cR u_2(t)\|_Y\le
l_\mathcal{J}^\mathcal{R}\,\|u_1(t)-u_2(t)\|_X\\[2mm]
\quad+L_\mathcal{J}^\mathcal{R}\,\displaystyle\int_0^t	\|u_1(s)-u_2(s)\|_X\,ds\ \
\mbox{\rm for all} \ \ u_1,\,u_2\in
C(I;X),\ \ t\in \mathcal{J}.\nonumber
\end{array} \right.$

\bigskip
\noindent 
$(\cS)$\qquad\ $\left\{\begin{array}{l}
\cS:C(I;X)\to C(I;X)\ \ \mbox{and for any compact set}\\[0mm]	
\mathcal{J}\subset I,\ \mbox{there exists}\ l_\mathcal{J}^\mathcal{S}>0\ \mbox{and}\ L_\mathcal{J}^\mathcal{S}>0\   \mbox{such that}\\[3mm]
\ \|\cS u_1(t)-\cS u_2(t)\|_X\le
l_\mathcal{J}^\mathcal{S}\,\|u_1(t)-u_2(t)\|_X\\[2mm]
\quad+L_\mathcal{J}^\mathcal{S}\,\displaystyle\int_0^t	\|u_1(s)-u_2(s)\|_X\,ds\ \
\mbox{\rm for all} \ \ u_1,\,u_2\in
C(I;X),\ \ t\in \mathcal{J}.\nonumber
\end{array} \right.$

\bigskip
\noindent 
$(j)$\qquad\ $\left\{\begin{array}{l}
j:Y\times K\to\R \mbox{ is such that}\\[2mm]	
{\rm (a)}\ \ j(\eta,\cdot):K\to\R\ \mbox{is a convex, positively homogenous}\\ \quad\quad\mbox{Lipschitz continuous function, for any}\ \eta\in Y.\\[2mm]
{\rm (b)}\  \ \mbox{There exists} \ \alpha_j\ge 0 \ \mbox{such that}\\
\qquad j(\eta_1,v_2)- j(\eta_1,v_1)+j(\eta_2,v_1) - j(\eta_2,v_2)\le\alpha_j\|\eta_1-\eta_2\|_{Y}\|v_1-v_2\|_{X}\\
\qquad\quad \mbox{for all}\ \eta_1,\ \eta_2\in Y,\ v_1,\,v_2\in K.	
\end{array} \right.$

\bigskip
\noindent 
$(f)$\qquad\ $f\in C(I;X)$. 

\medskip
Examples of operators ${\cal R}$, ${\cal S}$ and functions $j$ which satisfy conditions
 $({\cal R})$, $({\cal S})$ and $(j)$, respectively, will be provided in Sections \ref{s5} and \ref{s6}, in the study of several models of contact. We also mention that a history-dependent operator satisfies conditions $({\cal R})$ (or, equivalently, condition
  $({\cal S})$) and, therefore, additional examples are provided in
  \cite[pages 36--37, 39]{SMBOOK}. Nevertheless, for the convenience of the reader, we present here the following examples.

  \begin{example}\label{2ex4n}
  	Consider the operator $\cR\colon C(I;X)\to C(I;X)$ defined by
  	\begin{equation*}
  	\cR u(t)=e^tu(t)+\int_0^t s\,u(s)\,ds
  	\ \ \mbox{\rm for all} \ \ u\in C(I;X),\ t\in I.
  	\end{equation*}
  	Then it is easy to see that ${\cal R}$ satisfies condition  $({\cal R})$  with  \[l_\mathcal{J}^\mathcal{R}=\displaystyle \max_{t\in\mathcal J}\,e^t\qquad{\rm  and}\qquad L_\mathcal{J}^\mathcal{R}= \displaystyle\max_{t\in\mathcal J}\,t.\]
  In addition, note that ${\cal R}$ is not an almost history-dependent  operator.
  \end{example}

\begin{example} Let $j:Y\times K\to\R$ be the function defined by  $j(\eta,v)=p(\eta)q(v)$, where $p:Y\to\R$ and $q:K\to\R$. Assume that  $p$ is a Lipschitz continuous function with Lipschitz constant  $L_1$ and $q$  is a convex positively homogeneous  Lipschitz continuous function with Lipschitz constant $L_2$. Then, is easy to see that $j$ satisfies condition  $(j)$ with $\alpha_j=L_1 L_2$.

\end{example}

We now extend the function 	 $j$ from $Y\times K$ to the whole product space  $Y\times X$ by introducing the function $J:Y\times X\to(-\infty,+\infty]$ defined by
\begin{equation}\label{Jn}
J(\eta, v)=\left\{\begin{array}{l}
j(\eta,v)\qquad\mbox{if}\quad v\in K,\\[2mm]	
+\infty	\qquad\quad{\rm if}\quad v\notin K
\end{array} \right.\qquad\forall\, \eta\in Y.
\end{equation}
Using assumptions $(K)$ and $(j)$ it is easy to see that for any $\eta\in Y$, $J(\eta,\cdot)$ is proper,  positively homogenous, convex, lower semicontinuous and,
moreover, $J(\eta,0_X)=0$.	 Denote by $C(\eta)$ the subdifferential of $J(\eta,\cdot)$ in  $0_X$, i.e.,
\begin{equation}\label{Cn}
C(\eta)=\partial J(\eta,0_X)=
\{ \, \xi \in X \ :\  J(\eta,v)\ge (\xi,v)_X 
\ \ \forall\, v \in X \, \}\ 
\end{equation}
and, for any $t\in I$, let
\begin{equation}\label{Ct}
C(\eta,t)=f(t)-C(\eta).
\end{equation}
Note that, using assumptions $(K)$, $(j)$ and $(f)$ it follows that for any $\eta\in X$ and $t\in I$ the set 
$C(\eta,t)$ is a nonempty closed convex subset of $X$.

With these notation, the inclusion problem we consider in this section is the following.

\begin{problem}\label{p1}Find a function $u:I\to X$ such that
\begin{equation}\label{i}
	-u(t)\in {\rm N}_{C(\mathcal{R} u(t),t)}(Au(t)+\cS u(t))\qquad\forall\,t\in I.
	\end{equation}
\end{problem}

In the study of Problem \ref{p1} we have the following existence and uniqueness result. 

\medskip
\begin{theorem}\label{t1} 
Assume $(K)$--$(f)$ and, moreover, assume that for any compact set $\mathcal{J}\subset I$ the following smallness assumption holds:
\begin{equation}\label{smal}
(\alpha_j+1)(l_\mathcal{J}^\mathcal{R}+l_\mathcal{J}^\mathcal{S})<m_A.
\end{equation}Then, Problem $\ref{p1}$ has a unique solution with regularity $u\in C(I;K)$.	
\end{theorem}

 Before providing the proof of Theorem \ref{t1} we  start with a preliminary result which will repeatedly used in Sections \ref{s5} and \ref{s6} of this paper.
 
 \begin{lemma}\label{l0}
 	Let $X$, $Y$ be Hilbert spaces and assume  that $(K)$ and $(j)(a)$ hold.
 	Moreover, let $f:I\to X$, $\eta\in Y$, $u,\, z\in X$, $t\in I$ and let $J$, $C(\eta)$, $C(\eta,t)$ be given by $(\ref{Jn})$, $(\ref{Cn})$ and $(\ref{Ct})$, respectively. Then, the following equivalence holds:
 	\begin{equation}\label{Cm}
 u\in K,\ \	j(\eta,v)-j(\eta,u)\ge (f(t)-z,v-u)_X\ \ \forall\, v\in K\ \Longleftrightarrow\ -u\in {\rm N}_{C(\eta,t)}(z).
 	\end{equation}
 \end{lemma}
 
 \begin{proof}
 Using (\ref{Jn}) and the definition of the subdifferential have the equivalences
 \begin{eqnarray*}
 && u\in K,\quad j(\eta,v)-j(\eta,u)\ge (f(t)-z,v-u)_X\quad\forall\, v\in K\\ [2mm]
 &&\Longleftrightarrow\ J(\eta,v)-J(\eta,u)\ge (f(t)-z,v-u)_X\quad \forall\, v\in X \\ [2mm] &&\Longleftrightarrow\ f(t)-z\in \partial J(\eta,u)
 \end{eqnarray*}
 and, therefore, (\ref{e}) yields
 \begin{eqnarray}
 &&\label{e0}
 u\in K,\quad j(\eta,v)-j(\eta,u)\ge (f(t)-z,v-u)_X\quad\forall\, v\in K\ \\ [2mm] &&\Longleftrightarrow\ u\in \partial J^*(\eta,f(t)-z).\nonumber
 \end{eqnarray}
 
 On the other hand, assumption $(j)$ guarantees that $J(\eta,\cdot):X\to(-\infty,+\infty]$ is positively homogenuous with $J(\eta,0_X)=0$ and, therefore,
 $J(\eta,\cdot)=I^*_{C(\eta)}(\cdot)$  which implies that  $J^*(\eta,\cdot)=I^{**}_{C(\eta)}(\cdot)={\rm I}_{C(\eta)}(\cdot)$. It follows from here that  
 $\partial J^*(\eta,\cdot)={\rm N}_{C(\eta)}(\cdot)$. We use this equality to see that
  \begin{equation}\label{e00}
   u\in \partial J^*(\eta,f(t)-z)\ \Longleftrightarrow\ u\in {\rm N}_{C(\eta)}(f(t)-z).
  \end{equation}
 
 Finally, using (\ref{n1}) and  (\ref{n2}) we deduce that
 \begin{equation}\label{e000}
 u\in {\rm N}_{C(\eta)}(f(t)-z)= {\rm N}_{C(\eta)-f(t)}(-z) \ \Longleftrightarrow\ -u\in {\rm N}_{f(t)-C(\eta)}(z)
 \end{equation}
 We now combine the equivalences (\ref{e0})--(\ref{e000}), then we use notation (\ref{Ct}) to deduce that (\ref{Cm}) holds, which concludes the proof.
 \end{proof}

We now return back to the proof of Theorem \ref{t1} which is carried out in several steps that we describe in what follows. To this end, everywhere below we assume that  $(K)$--$(f)$ and (\ref{smal}) hold. The first step of the proof is the following.

\begin{lemma}\label{l1} For any $\theta=(\eta,\xi)\in C(I;Y\times X)$ there exists a unique function $u_\theta\in C(I;K)$ such that
\begin{equation}\label{i1}
-u_\theta(t)\in {\rm N}_{C(\eta(t),t)}(Au_\theta(t)+\xi(t))\qquad\forall\,t\in I.
\end{equation}
Moreover, if $u_i\in C(I;K)$ represents the solution of inclusion $(\ref{i1})$ for $\theta_i=(\xi_i,\eta_i)\in C(I;Y\times X)$, $i=1,2$, then
\begin{equation}\label{e1}
\|u_1(t)-u_2(t)\|_X\le\frac{1}{m_A}(\alpha_j\|\eta_1(t)-\eta_2(t)\|_Y+\|\xi_1(t)-\xi_2(t)\|_X)\quad\forall\,t\in I.
\end{equation}
\end{lemma}

\begin{proof} Let $\theta=(\eta,\xi)\in C(I;Y\times X)$. We use Lemma \ref{l0} to see that the time-dependent inclusion (\ref{i1}) is equivalent with the problem of finding a function $u_\theta:I\to X$ such that
\begin{eqnarray}
&&\label{11}\hspace{-15mm}u_\theta(t)\in K,\quad	j(\eta(t),v)-j(\eta(t),u_\theta(t))\ge (f(t)-Au_\theta(t)-\xi(t),v-u_\theta(t))_X\\ [2mm]
&&\hspace{55mm}\forall\, v\in K,\ t\in I.\nonumber
\end{eqnarray}
We claim that this time-dependent variational inequality has a unique solution $u_\theta\in C(I;K)$. To this end we consider an arbitrary element
$t\in I$ be fixed. Then, using assumptions $(K)$, $(A)$, $(j)$
it follows from Theorem~\ref{t0} that there exists a unique element $u_\theta(t)$ which solves
(\ref{11}). Now, let us prove that the map $t\mapsto u_\theta(t)\colon I\to K$ is continuous.
For this, consider $t_1$, $t_2\in I$ and, for the sake of simplicity in writing, denote
$\eta(t_i)=\eta_i$, $\xi(t_i)=\xi_i$, $u_{\theta}(t_i)=u_i$,  $f(t_i)=f_i$
for $i=1$, $2$. Using (\ref{11}) we obtain
\begin{eqnarray}
&&\hspace{-9mm}\label{2wv1}u_1\in K,\quad j(\eta_1,v)-j(\eta_1,u_1)\ge (f_1-Au_1-\xi_1,v-u_1)_X
\quad \forall\, v\in K,\\[2mm]
&&\hspace{-9mm}\label{2wv2}u_2\in K,\quad j(\eta_2,v)-j(\eta_2,u_2)\ge (f_2-Au_2-\xi_2,v-u_2)_X
\quad\forall\, v\in K.
\end{eqnarray}
Taking $v=u_2$ in (\ref{2wv1}), $v=u_1$ in (\ref{2wv2}) and adding the resulting
inequalities yields 
\begin{eqnarray}
&&\label{2w14n}(Au_1-Au_2,u_1-u_2)_X
\\ [1mm]
&&\quad\le j(\eta_1,u_2)-j(\eta_1,u_1)+j(\eta_2,u_1)-j(\eta_2,u_2)\nonumber\\ [1mm]
&&\qquad+(\xi_1-\xi_2,u_1-u_2)_X+
(f_1-f_2,u_1-u_2)_X.\nonumber
\end{eqnarray}
Then, using assumptions $(A)$ and $(j)(b)$, we obtain 
\begin{equation}
\label{2w16}
\hspace{-8mm}m_A\,\|u_1-u_2\|_X\\ \le \alpha_j\|\eta_1-\eta_2\|_Y+
\|\xi_1-\xi_2\|_X+\|f_1-f_2\|_X.
\end{equation}
Inequality (\ref{2w16}) combined with  assumption $(f)$ implies that $t\mapsto u_\theta(t)\colon 
I\rightarrow K$ is a continuous function. This concludes the existence part of the claim. The uniqueness part is a direct consequence of the
uniqueness of the solution $u_\theta(t)$ to the inequality (\ref{11}), at each $t\in I$, guaranteed by Theorem~\ref{t0}. 

Assume now that if $u_i\in C(I;K)$ represents the solution of inequality $(\ref{11})$ for $\theta_i=(\xi_i,\eta_i)\in C(I;Y\times X)$, $i=1,2$. Then, arguments similar to those used in the proof of inequality (\ref{2w16}) show that (\ref{e1}) holds.
Lemma \ref{l1} is now a direct conclusion of the equivalence between inclusion (\ref{i1}) and the inequality (\ref{11}), as already mentioned at the beginning of the proof.
\end{proof}

Next, we consider the operator $\Lambda:C(I;Y\times X)\to C(I;Y\times X)$ defined by
\begin{equation}\label{la}
\Lambda\theta=(\cR u_\theta,\cS u_\theta)\qquad\forall\, \theta\in C(I;Y\times X).
\end{equation}

We have the following result.

\begin{lemma}\label{l2} The operator $\Lambda$ has a unique fixed point $\theta^*=(\eta^*,\xi^*)\in C(I;Y\times X)$.
\end{lemma}

\begin{proof}
	Let $\theta_1=(\eta_1,\xi_1)$, $\theta_2=(\eta_2,\xi_2)\in C(I;Y\times X)$ and  denote by $u_i$ the solution of the
	variational inequality (\ref{11}) for $\theta=\theta_i$, i.e., $u_i=u_{\theta_i}$,
	$i=1$, $2$. Let $\mathcal{J}$ be a compact subset of $I$ 
	and $t\in \mathcal{J}$. 
	Then, using (\ref{la}) and assumptions $(\cR)$ and $(\cS)$ on the operators $\cR$ and
	$\cS$ yields
	\begin{eqnarray*}
	&&\hspace{-6mm}\|\Lambda\theta_1(t)-\Lambda\theta_2(t)\|_{Y\times X}\le \|\cR u_1(t)-\cR u_2(t)\|_Y+\|\cS u_1(t)-\cS u_2(t)\|_X\nonumber\\[2mm]
	&&\quad\le(l_\mathcal{J}^\mathcal{R}+l_\mathcal{J}^\mathcal{S})\|u_1(t)-u_2(t)\|_X+(L_\mathcal{J}^\mathcal{R}+L_\mathcal{J}^\mathcal{S})\int_0^t\|u_1(s)-u_2(s)\|_X\,ds.
	\end{eqnarray*}
	This inequality combined with inequality (\ref{e1}) and the elementary inequalities $\|\eta\|_Y\le\|\theta\|_{Y\times X}$, 
	$\|\xi\|_X\le\|\theta\|_{Y\times X}$, valid for all $\theta=(\eta,\xi)\in Y\times X$, implies that
	\begin{eqnarray*}
		&&\hspace{-6mm}\|\Lambda\theta_1(t)-\Lambda\theta_2(t)\|_{Y\times X}\ \le\frac{(\alpha_j+1)(l_\mathcal{J}^\mathcal{R}+l_\mathcal{J}^\mathcal{S})}{m_A}\|\theta_1(t)-\theta_2(t)\|_{Y\times X}\\[2mm]
		&&\quad+\frac{(\alpha_j+1)(L_\mathcal{J}^\mathcal{R}+L_\mathcal{J}^\mathcal{S})}{m_A}\int_0^t\|\theta_1(s)-\theta_2(s)\|_{Y\times X}\,ds.
	\end{eqnarray*}
We now use the smallness assumption (\ref{smal}) to obtain that the operator
	$\Lambda$ is an almost history-dependent operator, see Definition \ref{d} (b). Finally, we apply  Theorem \ref{t00} to conclude the proof of the lemma.
\end{proof}

We are now in a position to provide the proof of Theorem \ref{t1}.

\begin{proof} Let $\theta^*=(\eta^*,\xi^*)\in C(I;Y\times X)$ be the fixed point of the operator
$\Lambda$ and let $u^*=u_{\theta^*}\in C(I;K)$ be the solution of the intermediate problem (\ref{i1}) for $\theta=\theta^*$. Then, using
equality $\theta^*=\Lambda\theta^*$ we find that $\eta^*=\cR u^*$ and $\xi^*=\cS u^*$. We now use these equalities in (\ref{i1})
to see that $u^*$ is a solution to 	Problem \ref{p1}. This proves the existence part in Theorem \ref{t1}. The uniqueness part is a consequence of the uniqueness of the fixed point of the operator $\Lambda$, guaranteed  by Lemma \ref{l2}.	
\end{proof}

We end this sections with some consequence of Theorem \ref{t1} which are relevant for the applications we present in Section \ref{s5} of this paper.

\begin{corollary}\label{cor1}
Assume $(K)$, $(A)$, $(j)$, $(f)$ and, moreover, assume that $\cR:C(I;X)\to C(I;Y)$ and $\cS:C(I;X)\to C(I;X)$ are history-dependent  operators. Then, Problem $\ref{p1}$ has a unique solution with regularity $u\in C(I;K)$.	
\end{corollary}	
	
\begin{proof} Definition \ref{d} (a) shows that in this case conditions $({\cR})$ and $({\cS})$ are satisfied with $l^\mathcal{R}_\mathcal{J}=l^\mathcal{S}_\mathcal{J}=0$ and, therefore, the smallness condition (\ref{smal}) is satisfied. Corollary \ref{cor1} is now a direct consequence of Theorem \ref{t1}.
\end{proof}

\begin{corollary}\label{cor2}
Assume $(K)$, $(A)$, $(f)$ and, moreover, assume that $\cS:C(I;X)\to C(I;X)$ is a history-dependent operator.
In addition, assume that $j$ satisfies condition $(j)$ with $Y=X$ and
		\begin{equation}\label{sm}
		\alpha_j+1<m_A.
		\end{equation}Then, there exists a unique function $u\in C(I;K)$ such that	
		\begin{equation}\label{iz}
		-u(t)\in {\rm N}_{C(u(t),t)}(Au(t)+\cS u(t))\qquad\forall\,t\in I.
		\end{equation}
	\end{corollary}
\begin{proof} We take $\cR u=u$ for all $u\in C(I;X)$. Then, using Definition \ref{d} (a) we see that in this case conditions  $({\cR})$ and $({\cS})$ are satisfied with $l_\mathcal{J}^\mathcal{R}=1$ and $l^\mathcal{S}_\mathcal{J}=0$, respectively. Therefore, (\ref{sm}) implies that the smallness condition (\ref{smal}) holds, too. Corollary \ref{cor2} is now a direct consequence of Theorem \ref{t1}.
\end{proof}

We now consider the particular case when the function $j$ does not depend on the first variable, i.e. 
$j:K\to\real$. In this case we define the function $J:X\to(-\infty,+\infty]$ and the sets $C$, $C(t)\subset H$ by equalities
\begin{equation}\label{Jna}
J(v)=\left\{\begin{array}{l}
j(v)\qquad\mbox{if}\quad v\in K,\\[2mm]	
+\infty	\qquad{\rm if}\quad v\notin K,
\end{array} \right.
\end{equation}
\begin{equation}\label{Cnt}
C=\partial J(0_X), \qquad C(t)=f(t)-C\qquad\forall\, t\in I.
\end{equation}

With these notation, we have the following result which, clearly, represent a direct consequence of Theorem \ref{t1}.

\begin{corollary}\label{cor2n}
	Assume $(K)$, $(A)$, $(f)$ and, moreover, assume that $\cS:C(I;X)\to C(I;X)$ is a history-dependent operator.
	In addition, assume that $j:K\to\real$ is a convex positively homogenous Lipschitz continuous function. Then, there existe a unique function $u\in C(I;K)$ such that	
	\begin{equation*}\label{izm}
	-u(t)\in {\rm N}_{C(t)}(Au(t)+\cS u(t))\qquad\forall\,t\in I.
	\end{equation*}
\end{corollary}

Corollary \ref{cor2n} will be used in Section \ref{s4} in the study of a frictionless unilateral contact problem.

\section{Sweeping processes}\label{s4}

In this section we use Theorem \ref{t1} and its consequences in order to obtain existence and uniqueness results for several sweeping processes. To this end, besides the data $K$, $A$, $\cR$, $\cS$ $j$ and $f$ introduced in the previous section, we consider an operator $B$ and an initial data $u_0$ such that

\bigskip
\noindent 
$(B)$\qquad	   $B:X\to X$ is a Lipschitz continuous operator.

\medskip
\noindent 
$(u_0)$\qquad	  $u_0\in X$.

\medskip
We start by considering the following sweeping process.

\begin{problem}\label{p2}Find a function $u:I\to X$ such that
	\begin{eqnarray}
	&&\label{in}
	-\dot{u}(t)\in {\rm N}_{C(\mathcal{R} \dot{u}(t),t)}(A\dot{u}(t)+Bu(t)+\cS\dot{u}(t))\qquad\forall\,t\in I,\\ [2mm]
	&&\label{im}\ \ u(0)=u_0.
	\end{eqnarray}
\end{problem}

\medskip
Our first result in this section is the following.

\begin{theorem}\label{t2}	Assume $(K)$--$(f)$,  $(B)$, $(u_0)$ and, moreover, assume that $(\ref{smal})$ holds.
Then, Problem $\ref{p2}$ has a unique solution with regularity $u\in C^1(I;X)$ and  $\dot{u}\in C(I;K)$. 
\end{theorem}

\begin{proof} We introduce the operator $\widetilde{\cS}:C(I;X)\to C(I;X)$ defined by
\begin{equation}\label{ss1}
\widetilde{\cS}v(t)=B\big(\int_0^tv(s)\,ds+u_0\big)+ {\cS}v(t)
\end{equation}
for all $t\in I$, $v\in C(I;X)$, then we consider the auxiliary problem of finding a function 
$v:I\to X$ such that
\begin{equation}\label{ss2}
-v(t)\in {\rm N}_{C(\mathcal{R} v(t),t)}(Av(t)+\widetilde{\cS} v(t))\qquad\forall\,t\in I.
\end{equation}
Let $L_B$ be the Lipschitz constant of the operator $B$. We use assumptions $(\cS)$ and $(B)$ to see that	for any compact set
$\mathcal{J}\subset I$, any functions  $v_1,\,v_2\in
C(I;X)$ and any $t\in I$, the inequality below holds:
\begin{eqnarray*}
&&\ \|\widetilde{\cS}v_1(t)-\widetilde{\cS}v_2(t)\|_X\le
l_\mathcal{J}^\mathcal{S}\,\|v_1(t)-v_2(t)\|_X\\[2mm]
&&\qquad+(L_B+L_\mathcal{J}^\mathcal{S})\,\displaystyle\int_0^t	\|v_1(s)-v_2(s)\|_X\,ds.\nonumber
\end{eqnarray*}
It follows from here that the operator  $\widetilde{\cS}$ satisfies condition $(\cS)$ with $l_\mathcal{J}^{\widetilde{S}}=l_\mathcal{J}^\mathcal{S}$. Therefore, we are in a position to apply Theorem
\ref{t1} in order to obtain the existence of a unique function $v\in C(I;K)$ which satisfies the time-dependent inclusion (\ref{ss2}). Denote by $u:I\to X$
the function defined by
\begin{equation}\label{ss3}
u(t)=\int_0^t	v(s)+u_0\qquad\forall\, t\in I.
\end{equation}
Then, (\ref{ss1})--(\ref{ss3}) and assumption $(u_0)$ imply that $u$ is a solution of Problem \ref{p2} with regularity $u\in C^1(I;X)$ and $\dot{u}\in C(I;K)$. This proves the existence part of the theorem. The uniqueness part follows from the unique solvability of the auxiliary problem (\ref{ss2}), guaranteed by Theorem \ref{t1}.
\end{proof}	

Theorem \ref{t2} can be used in the study of various versions of sweeping process of the form (\ref{in}) and (\ref{im}). We provide below some consequence of this theorem in the study of three relevant examples.

\medskip

\begin{corollary}\label{cor3n}
	Assume $(K)$, $(A)$, $(j)$, $(f)$, $(B)$, $(u_0)$ and, moreover, assume that $\cR:C(I;X)\to C(I;Y)$ and $\cS:C(I;X)\to C(I;X)$ are history-dependent operators. 
	Then, Problem $\ref{p2}$ has a unique solution with regularity $u\in C^1(I;X)$ and $\dot{u}\in C(I;K)$. 
\end{corollary}

\begin{proof} Definition \ref{d} (a) shows that in this case conditions $({\cR})$ and $({\cS})$ are satisfied with $l^\mathcal{R}_\mathcal{J}=l^\mathcal{S}_\mathcal{J}=0$ and, therefore, the smallness condition (\ref{smal}) is satisfied. Corollary \ref{cor2n} is now a direct consequence of Theorem \ref{t2}.
\end{proof}	

\medskip

\begin{corollary}\label{cor3}
	Assume $(K)$, $(A)$, $(f)$, $(B)$, $(u_0)$ and, moreover, assume that $\cS:C(I;X)\to C(I;X)$ is a history-dependent operator. In addition, assume that $j$ satisfies condition $(j)$ with $Y=X$.
	Then, there existe a unique function $u\in C^1(I;X)$  such that	
	\begin{eqnarray*}
	&&\label{izn}
	-\dot{u}(t)\in {\rm N}_{C(u(t),t)}(A\dot{u}(t)+Bu(t)+\cS\dot{u}(t))\qquad\forall\,t\in I,\\ [2mm]
	&&\ \ u(0)=u_0.
	\end{eqnarray*}
Moreover, $\dot{u}\in C(I;K)$.
\end{corollary}	

\begin{proof} Consider the operator $\cR: C(I;X)\to C(I;X)$ defined by equality
\begin{equation*}\label{ss4}
	\cR v(t)=\int_0^t v(s)\,ds+u_0\qquad\forall\, v\in C(I;V),\ t\in I.
\end{equation*}
Then, using Definition \ref{d} (a) we see that in this case conditions  $({\cR})$ and $({\cS})$ are satisfied with $l_\mathcal{J}^\mathcal{R}=0$ and $l^\mathcal{S}_\mathcal{J}=0$, respectively.
Therefore, the smallness condition (\ref{smal}) is satisfied. Moreover,
${\cR}\dot{u}=u$ for all $u\in C(I;X)$. Corollary \ref{cor3} is now a direct consequence of Corollary \ref{cor3n}.
\end{proof}

\medskip
\begin{corollary}\label{cor4}
	Assume $(K)$,  $(A)$, $(f)$,  $(B)$,  $(u_0)$, and, moreover, assume that $\cS:C(I;X)\to C(I;X)$ is a history-dependent operator.
	In addition, assume that $j$ satisfies condition $(j)$ with $Y=X$.
Then, there existe a unique function $u\in C(I;K)$ such that	
	\begin{eqnarray}
	&&\label{iza}
	-\dot{u}(t)\in {\rm N}_{C(u(t),t)}(A\dot{u}(t)+Bu(t)+\cS u(t))\qquad\forall\,t\in I,\\ [2mm]
	&&\ \ u(0)=u_0.
	\end{eqnarray}
	Moreover, $\dot{u}\in C(I;K)$.
\end{corollary}

\begin{proof} Consider the operator $\widetilde{\cS}: C(I;X)\to C(I;X)$ defined by equality
	\begin{equation}\label{ss4n}
	\widetilde{\cS}v(t)=\cS\big(\int_0^tv(s)\,ds+u_0\big)\qquad\forall\, v\in C(I;V),\ t\in I.
	\end{equation}
Then, using Definition \ref{d} (a) it is easy to see that $\widetilde{\cS}$ is a history-dependent operator and, moreover,
	$\widetilde{\cS}\dot{u}=\cS u$ for all $u\in C^1(I;X)$.
	Corollary \ref{cor4} is now a direct consequence of Corollary \ref{cor3}.
\end{proof}

\section{Two frictionless contact problems}\label{s5}

The physical setting, already considered in many papers and surveys, can be resumed as follows. A deformable body occupies, in its reference configuration, a
bounded domain $\Omega\subset\mathbb{R}^d$ ($d=1,2,3$), with a
Lipschitz continuous boundary $\Gamma$, divided into three
measurable disjoint parts $\Gamma_1$, $\Gamma_2$ and $\Gamma_3$, such that ${ meas}\,(\Gamma_1)>0$.  The body is fixed on $\Gamma_1$, is acted upon by given surface tractions on $\Gamma_2$, and is in contact
with an obstacle on $\Gamma_3$. The equilibrium of the  body in this physical setting can be described by various mathematical models, obtained by using different mechanical assumptions. The  first contact model we consider 
in this section is based on specific constitutive law and interface boundary conditions which will be described below. Its statement  is as follows.

\medskip
\begin{problem}\label{p1m} Find a displacement field
	$\bu \colon \Omega\times I\to\mathbb{R}^d$
	and a  stress field $\bsigma \colon \Omega\times I\to\mathbb{S}^d$
	such that
	\begin{align}
	\label{1m} \bsigma(t)={\cal A} \bvarepsilon({\bu}(t))+\int_0^t{\cal B}(t-s)&\bvarepsilon({\bu}(s))\,ds\quad&{\rm in}\
	&\Omega,\\[3mm]
	\label{2m} {\rm Div}\,\bsigma(t)+\fb_0(t)&=\bzero\quad&{\rm in}\ &\Omega,\\[2mm]
	\label{3m} \bu(t)&=\bzero &{\rm on}\ &\Gamma_1,\\[2mm]
	\label{4m} \bsigma(t)\bnu&=\fb_2(t)\quad&{\rm on}\ &\Gamma_2,\\[4mm]
	\label{5m}
	\hspace{0mm}\left.\begin{array}{lll}
	-F\Big(\displaystyle\int_0^tu_\nu^+(s)\,ds\Big)\le\sigma_\nu(t)\le0,\\[6mm] -\sigma_\nu(t)=
	\left\{\begin{array}{ll}0\quad{\rm if}\quad u_\nu(t)<0,\\[4mm]
	F\Big(\displaystyle\int_0^tu_\nu^+(s)\,ds\Big)
	\quad{\rm if}\quad u_\nu(t)>0,\\[2mm]
	\end{array}\right.\\[11mm]
	\end{array}\right\}\hspace{-28mm}&\ 
	&{\rm on}\ &\Gamma_3, \\[4mm]
	\label{6m} \bsigma_\tau(t)&=\bzero
	\quad&{\rm on}\ &\Gamma_3
	\end{align}
	for all $t\in I$.
\end{problem}

Here and below, in order to simplify the notation, we do not
indicate explicitly the dependence of various functions on the
spatial variable $\bx\in\Omega\cup\Gamma$.  Moreover, we use the notation introduced  in Section \ref{s2} and, in addition, $\sigma_\nu$ and $\bsigma_\tau$ denote  the
normal and tangential stress on $\Gamma$, that is
$\sigma_{\nu}=(\bsigma\bnu)\cdot\bnu$ and $\bsigma_{\tau} =
\bsigma\bnu - \sigma_{\nu}\bnu$. We now provide a short description of the equations and boundary conditions in Problem \ref{p1m}.

First, equation (\ref{1m}) represents the 
constitutive law  in which ${\cal A}$  is the elasticity operator, assumed to be nonlinear, and ${\cal B}$ represents the relaxation tensor. Next, equation (\ref{2m}) is
the equation of equilibrium in which ${\fb}_0$ represents the density of the body forces, assumed to be time-dependent.
Condition
(\ref{3m}) represents the displacement boundary condition which 
shows that the body is fixed on the part $\Gamma_1$ of its boundary, 
during the process. Condition (\ref{4m}) represents the  traction
condition which shows that surface tractions of density ${\fb}_2$, 
assumed to be time-dependent, act on $\Gamma_2$.
Condition (\ref{5m})  
models the contact with a rigid-deformable body with memory effects. Here
$F$ is a positive function and $r^+$ represents the positive part of $r$, i.e., $r^+ = \max\,\{r,0\}$.
Details on this condition can be found in \cite[Ch.9]{SMBOOK}. Finally, condition (\ref{6m}) represents the
frictionless contact condition. It shows that the friction force, $\bsigma_\tau$, vanishes during the process. This is an idealization of the process, since even completely
lubricated surfaces generate shear resistance to tangential
motion. However, this condition  is a
sufficiently good approximation of the reality in some situations, especially when the contact surfaces are lubricated.

\medskip

In the study of the mechanical problem (\ref{1m})--(\ref{6m}) we
assume that the elasticity operator ${\cal A}$ satisfies the
following conditions.

\begin{equation}
\left\{ \begin{array}{ll} {\rm (a)\ } {\cal A}:\Omega\times
\mathbb{S}^d\to \mathbb{S}^d.\\[2mm]
{\rm (b)\  There\ exists}\ L_{\cal A}>0\ {\rm such\ that}\\
{}\qquad \|{\cal A}(\bx,\bvarepsilon_1)-{\cal A}(\bx,\bvarepsilon_2)\|
\le L_{\cal A} \|\bvarepsilon_1-\bvarepsilon_2\|\\
{}\qquad\qquad\forall\,\bvarepsilon_1,\bvarepsilon_2
\in \mathbb{S}^d,\ {\rm a.e.}\ \bx\in \Omega.\\[2mm]
{\rm (c)\  There\ exists}\ m_{\cal A}>0\ {\rm such\ that}\\
{}\qquad ({\cal A}(\bx,\bvarepsilon_1)-{\cal A}(\bx,\bvarepsilon_2))
\cdot(\bvarepsilon_1-\bvarepsilon_2)\ge m_{\cal A}\,
\|\bvarepsilon_1-\bvarepsilon_2\|^2\\
{}\qquad\quad \forall\,\bvarepsilon_1,
\bvarepsilon_2 \in \mathbb{S}^d,\ {\rm a.e.}\ \bx\in \Omega.\\[2mm]
{\rm (d)\ The\ mapping}\ \bx\mapsto
{\cal A}(\bx,\bvarepsilon)\ {\rm is\ measurable\ on\
}\Omega,\\
{}\qquad {\rm for\ any\ }\bvarepsilon\in \mathbb{S}^d.\\[2mm]
{\rm (e)}\ {\cal A}(\bx,\bzero)=\bzero \ \ {\rm a.e.}\ \bx\in \Omega.
\end{array}\right.
\label{A}
\end{equation}

We also assume that the relaxation tensor ${\cal B}$ and the densities of  body forces and surface tractions are such that
\begin{eqnarray}
&&\label{B} {\cal B}\in C(I;{\bf Q}_\infty). \\
&&\label{f0} \fb_0 \in C(I;L^2(\Omega)^d).\\
&&\label{f2}\fb_2\in C(I;L^2(\Gamma_2)^d). 
\end{eqnarray}
Finally, the memory surface function $F$ satisfies:
\begin{equation}\left\{\begin{array}{ll}
	F\colon \Gamma_3\times\mathbb R\to \mathbb R_+.\\ [1mm]
	{\rm (a)\ } \mbox{There exists }L_{F}>0\mbox{ such that }\\
	\qquad |F(\bx,r_1)-F(\bx,r_2)|\leq L_{F}|r_1-r_2|\\
	\qquad\quad\forall\,
	r_1, r_2\in {\mathbb R},\  {\rm a.e.\   } \bx\in\Gamma_3, \\ [1mm]
	{\rm (b) }\ \mbox{The mapping } \bx\mapsto F(\bx, r){\rm\ is\ measurable\ on\ }\Gamma_3 \ {\rm for\ any\  } r\in \mathbb{R},\\ [1mm]
	{\rm (c)}\ 
	F(\bx,0)= 0 \ \ {\rm a.e.\   } \bx\in\Gamma_3.
\end{array}\right.\label{F}
\end{equation}

We now turn to the variational formulation of  Problem \ref{p1m} and, to this end, we assume in what follows that
$(\bu,\bsigma)$ represents a couple of regular functions which satisfies (\ref{1m})--(\ref{6m}).  
Then, using  standard 
arguments based on the Green formula (\ref{Green}) we find that
\begin{eqnarray}
&&\hspace{-8mm}\label{7m}
\int_\Omega\bsigma(t)\cdot
(\bvarepsilon(\bv)-\bvarepsilon(\bu(t)))\,dx\\[2mm]
&&\hspace{-8mm}\quad+  \int_{\Gamma_3}F\Big(\int_0^tu_\nu^+(s)\,ds\Big)( v_\nu^+ - u_\nu^+(t)) \,da
+\nonumber\\[2mm]
&&\hspace{-8mm}\qquad
\ge \int_\Omega\fb_0(t)\cdot
(\bv-\bu(t))\,dx+\int_{\Gamma_2}\fb_2(t)\cdot
(\bv-\bu(t))\,da\nonumber
\end{eqnarray}
for all $\bv \in V$ and every $t\in I$. Recall that here and in the rest of the paper we use the function spaces $V$ and $Q$ introduced in Section \ref{s2}.
We now consider the operators $A\colon V \to V$, $\cR \colon C(I ; V)\to C(I; L^2(\Gamma_3))$,
	$\cS \colon C(I ; V)\to C(I; V)$, 
	the functional
	$j\colon L^2(\Gamma_3)\times V\to\real$ and the function $\fb \colon I\to V$ defined by 
	\begin{eqnarray}
	&&\label{8m} 
	(A\bu,\bv)_V =\int_\Omega{\cal A}\bvarepsilon(\bu)\cdot\bvarepsilon(\bv)\,dx
	\quad \mbox{for all} \ \bu,\bv\in V,\\[2mm]
	&&\label{9m}\cR\bu(t)=F\Big(\int_0^t u_\nu^+(s)\,ds\Big)\quad \mbox{for all} \ \bu \in 
	C(I;V),\\ [2mm]
	&&\label{10m}(\cS\bu(t),\bvarepsilon(\bv))_V=(\int_0^t{\cal B}(t-s)\bvarepsilon(\bu(s))\,ds,\bvarepsilon(\bv))_Q\\ &&\qquad\qquad \mbox{for all} \ \bu, \bv\in 
	C(I;V),\nonumber
	\\[2mm]
	&&\label{11m}j(\eta,\bv)=
	\int_{\Gamma_3}\eta v_\nu^+\, da \quad \mbox{for all} \ \eta\in L^2(\Gamma_3),\ \bv\in V,\\
	[2mm]
	&&\label{12m}(\fb(t), \bv )_V=\int_\Omega\fb_0(t)\cdot\bv\,dx +
	\int_{\Gamma_2}\fb_2(t)\cdot\bv\,da\quad \mbox{for all} \ \bv\in V,\ 
	t\in I.
	\end{eqnarray}
	
	We now substitute equation (\ref{1m})  in (\ref{7m}), then we use notation  (\ref{8m})--(\ref{12m}) to see that
\begin{equation}\label{13m}
j(\cR\bu(t),\bv)-j(\cR\bu(t),\bu(t))\ge(\fb(t)-A\bu(t)-\cS\bu(t),\bv-\bu(t))_V
\end{equation}
for all $\bv \in V$ and every $t\in I$. Let
\begin{equation*}
C(\eta)=\partial j(\eta,\bzero_V),\qquad C(\eta,t)=\fb(t)-C(\eta)\qquad \mbox{for all} \ \eta\in L^2(\Gamma_3),\ 
t\in I.
\end{equation*}

We take $X=V$, $K=V$ and note that in this case condition $(K)$ is satisfied. Moreover, taking $Y=L^2(\Gamma_3)$ and using the trace inequality (\ref{trace}) it is easy to see that
condition $(j)(a)$ is satisfied, too.
Therefore, from inequality (\ref{13m}) and Lemma \ref{l0} with $J=j$,
we derive the following variational formulation of Problem~\ref{p2}.
\begin{problem}\label{p2m}
	Find a displacement field $\bu \colon I\to V$  such that
\begin{equation}\label{imm}
-\bu(t)\in {\rm N}_{C(\mathcal{R} \bu(t),t)}(A\bu(t)+\cS \bu(t))\qquad\forall\,t\in I.
\end{equation}	
\end{problem}

In the study of  Problem~\ref{p2m} we have the following existence and uniqueness result. 
\begin{theorem}\label{t1m}
	Assume that $(\ref{A})$--$(\ref{F})$ hold.
	Then Problem~$\ref{p2m}$ has a unique solution $\bu \in  C(I;V)$.
\end{theorem}

\begin{proof} We use Corollary \ref{cor1} on the spaces $X=V$, $Y=L^2(\Gamma_3)$, with $K=V$. As already mentioned, assumptions $(K)$ and $(j)(a)$ are obviously satisfied. Moreover, using (\ref{11m}) and (\ref{trace}) it is easy to see that for any $\eta_1,\, \eta_2\in L^2(\Gamma_3)$ and any $\bu_1,\ \bu_2\in V$ we have
\[j(\eta_1,\bu_2)-j(\eta_1,\bu_1)+j(\eta_2,\bu_1)-j(\eta_2,\bu_2)\le c_0\|\eta_1-\eta_2\|_{ L^2(\Gamma_3)}\|\bu_1-\bu_2\|_V,\]
which implies that function $j$ satisfies condition $(j)(b)$  with $\alpha_j=c_0$. On the other hand, assumption
(\ref{A}) implies that for any $\bu,\, \bv\in V$ the inequalities below hold:
\begin{eqnarray*}
&&(A\bu-A\bv,\bu-\bv)_V\ge m_{\cal A}\|\bu-\bv\|_V^2,\\ [2mm]
&&\|A\bu-A\bv\|_V\le L_{\cal A}\|\bu-\bv\|_V.
\end{eqnarray*}
We conclude from here that condition $(A)$ is satisfied. Next, we use assumptions (\ref{F}), (\ref{B}) and inequalities (\ref{trace}), (\ref{pmp}) to see that for any compact $\mathcal{J}$, any functions
$\bu_1$, $\bu_2$ and any $t\in \mathcal{J}$ we have
\begin{eqnarray*}
&&\|\cR\bu_1(t)-\cR\bu_2(t)\|_{L^2(\Gamma_3)}\le c_0L_F\,\displaystyle\int_0^t	\|\bu_1(s)-\bu_2(s)\|_Y\,ds,\ \\ [2mm]
&&\|\cS\bu_1(t)-\cS\bu_2(t)\|_V \le d\,\max_{s\in \mathcal{J}}\|{\cal B}(s)\|_{\bf Q_\infty}\,\displaystyle\int_0^t	\|\bu_1(s)-\bu_2(s)\|_V\,ds,\ \
\end{eqnarray*} 	
which prove that the operators $\cR$ and $\cS$ are history-dependent operators. Finally, the regularities  (\ref{f0}) and (\ref{f2}) imply that
$\fb\in C(I;V)$ and, therefore, condition $({f})$ holds, too.	
Theorem \ref{t1m} is now  direct consequence of Corollary \ref{cor1}.
\end{proof}

A second viscoelastic  contact problem for which the abstract results provided in Section \ref{s3} work is the Signorini frictionless contact problem, which models the contact with a perfectly rigid foundation. The statement of this problem is the following.

\medskip
\begin{problem}\label{p3m} Find a displacement field
	$\bu \colon \Omega\times I\to\mathbb{R}^d$
	and a  stress field $\bsigma\colon \Omega\times I\to\mathbb{S}^d$
	such that $(\ref{1m})$--$(\ref{4m})$, $(\ref{6m})$ hold
	for all $t\in I$ and, moreover,
	\begin{equation}\label{15mm}
	u_\nu(t)\le0,\quad\sigma_\nu(t)\le 0,\quad \sigma_\nu(t)u_\nu(t)=0\end{equation}
	for all $t\in I$.
\end{problem}

\medskip

We assume conditions (\ref{A})--(\ref{f2}) and use notation (\ref{8m}), (\ref{10m}) and (\ref{12m}). Moreover,  we consider
the set $U$ and  the function $j:U\to\real$
defined by 
\begin{eqnarray}
&&\label{Um}U=\{\,\bv\in V : v_\nu \le 0\ \  \hbox{a.e. on}\
\Gamma_3\,\},\\ [2mm]
&&\label{jp}j(\bv)=\bzero\qquad\forall\, \bv\in U.
\end{eqnarray}

Note that in this case the function $j$ does not depend on the first variable and, therefore, using notations (\ref{Jna}), (\ref{Cnt})
with  $X=V$, $K=U$ we deduce that $J={\rm I}_U$, $C={\rm N}_U(\bzero_V)$ and $C(t)=\fb(t)-{\rm N}_U(\bzero_V)$ for all $t\in I$.
Then, using arguments similar to those used in the study of Problem~\ref{p1m}, based on the Green formula and Lemma \ref{l0}, 
we derive the following variational formulation of Problem~\ref{p3m}.
\begin{problem}\label{p4m}
	Find a displacement field $\bu \colon I\to V$  such that
	\begin{equation}\label{a10}
	-\bu(t)\in {\rm N}_{C(t)}(A\bu(t)+\cS \bu(t))\qquad\forall\,t\in I.
	\end{equation}	
\end{problem}

In the study of  Problem~\ref{p4m} we have the following existence and uniqueness result. 

\medskip
\begin{theorem}\label{t2m}
	Assume that $(\ref{A})$--$(\ref{f2})$ hold.
	Then Problem~$\ref{p4m}$ has a unique solution $\bu \in  C(I;U)$.
\end{theorem}

\begin{proof} The proof of Theorem \ref{t2m} is a direct consequence of Corollary \ref{cor2n}. It is based on arguments similar  to those used in the proof of Theorem
\ref{t1m} and, for this reason, we skip the details.	
\end{proof}

\section{A frictional viscoelastic contact problem}\label{s6}

For the  model we consider in this section the contact is frictional. As a consequence, its variational formulation leads to a sweeping process in which the unknown is the displacement field. The model is formulated as follows.

\medskip
\begin{problem}\label{p5m} Find a displacement field
	$\bu \colon \Omega\times I\to\mathbb{R}^d$
    and a  stress field 	$\bsigma \colon \Omega\times I\to\mathbb{S}^d$
	such that
	\begin{align}
	\label{1mz} \bsigma(t)={\cal A} \bvarepsilon(\dot{\bu}(t))+{\cal E} \bvarepsilon({\bu}(t))+\int_0^t{\cal B}(t-s)&\bvarepsilon(\dot{\bu}(s))\,ds\quad&{\rm in}\
	&\Omega,\\[3mm]
	\label{2mz} {\rm Div}\,\bsigma(t)+\fb_0(t)&=\bzero\quad&{\rm in}\ &\Omega,\\[2mm]
	\label{3mz} \bu(t)&=\bzero &{\rm on}\ &\Gamma_1,\\[2mm]
	\label{4mz} \bsigma(t)\bnu&=\fb_2(t)\quad&{\rm on}\ &\Gamma_2,\\[4mm]
	\label{5mz}
	u_\nu(t)&=0\quad&{\rm on}\ &\Gamma_3,\\[4mm]
	\hspace{0mm}\left.\begin{array}{ll}
	\|\bsigma_\tau(t)\|\le F(\displaystyle\int_0^t\|\dot{\bu}_\tau(s)\|\,ds),\\
	[6mm] -\bsigma_\tau(t)=F(\displaystyle\int_0^t\|\dot{\bu}_\tau(s)\|\,ds)\frac{	\dot{\bu}_\tau(t)}{\|\dot{\bu}_\tau(t)\|}\quad{\rm if}\quad\dot{\bu}_\tau(t)\ne0
	\end{array}\right\}
	\hspace{-28mm}&\ 
	&{\rm on}\ &\Gamma_3
	\label{6mz}
	\end{align}
	for all $t\in I$ and, moreover,
	\begin{equation}	\label{7mz}
	\bu(0)=\bu_0\qquad\quad\ {\rm in}\ \ \Omega.
	\end{equation}
\end{problem}

The equations and boundary conditions in Problem \ref{p5m} have a similar meaning to those in Problems \ref{p1m} and  \ref{p3m} studied in the previous section. Note that (\ref{1mz}) represents the constitutive law in which now ${\cal A}$ represents the viscosity operator, ${\cal E}$ is the elasticity operator and, again, ${\cal B}$ represents the relaxation tensor. Condition (\ref{5mz}) represents the bilateral contact condition; it shows that there is no separation between the body and the foundation, during the process. Condition (\ref{6mz}) represents a total-slip version of Coulomb's law of dry friction. Here $F$ denotes the friction bound and the quantity
\[T(\bx,t)=\int_0^t\|\dot{\bu}_\tau(\bx,s)\|\,ds\]
represents the total slip-rate in the point $\bx\in\Gamma_3$, at the time moment $t\in I$. Considering a friction bound $F$ which depends on the total slip rate describes the rearrangement of the contact surfaces during the sliding process. Finally, condition (\ref{6mz}) represents the initial condition in which $\bu_0$ denotes a given initial displacement field.

The weak solution of the mechanical problem (\ref{1m})--(\ref{6m}) will be sought in  the space
\[
V_1=\{\,\bv\in V:\  v_\nu =0\ \ {\rm on\ \ }\Gamma_3\,\}.
\]
Note that $V_1$ is a closed subspace of the space $V$ and, therefore, is a Hilbert space equipped with the inner product $(\cdot,\cdot)_V$ and the associated norm $\|\cdot\|_V$.

\medskip

In the study of the mechanical problem (\ref{1mz})--(\ref{7mz}) we
assume that the viscosity operator ${\cal A}$ and the relaxation tensor satisfy conditions (\ref{A}) and  (\ref{B}), respectively. Moreover, the density of applied forces and the friction bound are such that   (\ref{f0}),  (\ref{f2}) and  (\ref{F}), hold. Finally, for the elasticity operator and the initial displacement we assume that
\begin{equation}
\left\{ \begin{array}{ll} {\rm (a)\ } {\cal E}:\Omega\times
\mathbb{S}^d\to \mathbb{S}^d.\\[2mm]
{\rm (b)\  There\ exists}\ L_{\cal E}>0\ {\rm such\ that}\\
{}\qquad \|{\cal E}(\bx,\bvarepsilon_1)-{\cal E}(\bx,\bvarepsilon_2)\|
\le L_{\cal E} \|\bvarepsilon_1-\bvarepsilon_2\|\\
{}\qquad\qquad\forall\,\bvarepsilon_1,\bvarepsilon_2
\in \mathbb{S}^d,\ {\rm a.e.}\ \bx\in \Omega.\\[2mm]
{\rm (c)\ The\ mapping}\ \bx\mapsto
{\cal E}(\bx,\bvarepsilon)\ {\rm is\ measurable\ on\
}\Omega,\\
{}\qquad {\rm for\ any\ }\bvarepsilon\in \mathbb{S}^d.\\[2mm]
{\rm (d)}\ {\cal E}(\bx,\bzero)=\bzero\  \ {\rm a.e.}\  \bx\in \Omega.
\end{array}\right.
\label{E}
\end{equation}

\begin{equation}\label{u0}
\bu_0\in V_1.
\end{equation}

\medskip
We now turn to the variational formulation of  Problem \ref{p5m} and, to this end, we assume in what follows that
$(\bu,\bsigma)$ represents a couple of regular functions which satisfies (\ref{1mz})--(\ref{7mz}).  
Then, using  standard 
arguments based on the Green formula (\ref{Green}) we find that

\begin{eqnarray}
&&\hspace{-8mm}\label{7mbz}
\int_\Omega\bsigma(t)\cdot
(\bvarepsilon(\bv)-\bvarepsilon(\dot{\bu}(t)))\,dx\\[2mm]
&&\hspace{-8mm}\quad+  \int_{\Gamma_3}F\Big(\displaystyle\int_0^t\|\dot{\bu}_\tau(s)\|\,ds\Big)( \|{\bv}_\tau(s)\|- \|\dot{\bu}_\tau(s)\|)\,da
\nonumber\\[2mm]
&&\hspace{-8mm}\qquad
\ge \int_\Omega\fb_0(t)\cdot
(\bv-\dot{\bu}(t))\,dx+\int_{\Gamma_2}\fb_2(t)\cdot
(\bv-\dot{\bu}(t))\,da\nonumber
\end{eqnarray}
for all $\bv \in V_1$ and every $t\in I$.
We now introduce the operators $A\colon V_1 \to V_1$, $ B\colon V_1 \to V_1$, $\cR \colon C(I ; V_1)\to C(I; L^2(\Gamma_3))$,
$\cS \colon C(I ; V_1)\to C(I; V_1)$, 
the functional
$j\colon L^2(\Gamma_3)\times V_1\to\real$ and the function $\fb \colon I\to V_1$ defined by 
\begin{eqnarray}
&&\label{8mz} 
(A\bu,\bv)_V =\int_\Omega{\cal A}\bvarepsilon(\bu)\cdot\bvarepsilon(\bv)\,dx
\quad \mbox{for all} \ \bu,\bv\in V_1,\\[2mm]
&&\label{10mbz}(B\bu,\bv)_V=\int_\Omega{\cal E}\bvarepsilon(\bu)\cdot\bvarepsilon(\bv)\,dx\quad \mbox{for all} \ \bu,\bv\in V_1,\\[2mm]
&&\label{9mz}\cR\bu(t)=F\Big(\int_0^t \|{\bu}_\tau(t)\|\,ds\Big)\quad \mbox{for all} \ \bu \in 
C(I;V_1),\\ [2mm]
&&\label{10mz}(\cS\bu(t),\bv)_V=(\int_0^t{\cal B}(t-s)\bvarepsilon(\bu(s))\,ds,\bvarepsilon(\bv))_Q\\ [2mm]
&&\qquad\qquad \mbox{for all} \ \bu, \bv\in 
C(I;V_1),\nonumber
\\[2mm]
&&\label{11mz}j(\eta,\bv)=
\int_{\Gamma_3}\eta\, \|{\bv}_\tau(t)\|\, da \quad \mbox{for all} \ \eta\in L^2(\Gamma_3),\ \bv\in V_1,\\
[2mm]
&&\label{12mz}(\fb(t), \bv)_V=\int_\Omega\fb_0(t)\cdot\bv\,dx +
\int_{\Gamma_2}\fb_2(t)\cdot\bv\,da\quad \mbox{for all} \ \bv\in V_1,\ 
t\in I.
\end{eqnarray}

We now substitute equation (\ref{1mz})  in (\ref{7mbz}), then we use notation  (\ref{8mz})--(\ref{12mz}) to see that
\begin{equation}\label{13mz}
j(\cR\dot{\bu}(t),\bv)-j(\cR\dot{\bu}(t),\dot{\bu}(t))\ge(\fb(t)-A\dot{\bu}(t)-B\bu(t)-\cS\dot{\bu}(t),\bv-\dot{\bu}(t))_V
\end{equation}
for all $\bv \in V_1$ and  $t\in I$. Let
\begin{equation*}\label{Ctpz}
C(\eta)=\partial j(\eta,\bzero_V),\qquad C(\eta,t)=\fb(t)-C(\eta)\qquad \mbox{for all} \ \eta\in L^2(\Gamma_3),\ 
t\in I.
\end{equation*}

Take $X=V_1$, $K=V_1$ and note that in this case conditions $(K)$ and $(j)(a)$ are satisfied, the later one being the consequence of the trace inequality (\ref{trace}).
Then, using inequality (\ref{13mz}), Lemma \ref{l0} with $J=j$ and the initial condition (\ref{7mz}),
we derive the following variational formulation of Problem~\ref{p5m}.
\begin{problem}\label{p6m}
	Find a displacement field $\bu \colon I\to V_1$  such that
	\begin{eqnarray*}
	&&-\dot{\bu}(t)\in {\rm N}_{C(\mathcal{R} \dot{\bu}(t),t)}(A\dot{\bu}(t)+B\bu(t)+\cS \dot{\bu}(t))\qquad\forall\,t\in I,\\ [2mm]
	&&\ \ \bu(0)=\bu_0.\label{z11}
	\end{eqnarray*}	
\end{problem}

In the study of  Problem~\ref{p6m} we have the following existence and uniqueness result. 
\begin{theorem}\label{t3m}
	Assume that $(\ref{A})$--$(\ref{F})$, $(\ref{E})$, $(\ref{u0})$ hold.
	Then Problem~$\ref{p6m}$ has a unique solution $\bu \in  C^1(I;V_1)$.
\end{theorem}

\begin{proof}  We use Corollary \ref{cor3n} on the spaces $X=V_1$, $Y=L^2(\Gamma_3)$, with $K=V_1$. As already mentioned, assumptions $(K)$ and $(j)(a)$ are obviously satisfied. Moreover, it follows from arguments similar to those in the proof of Theorem \ref{t1m}  that assumptions $(j)(b)$, $(A)$, $(f)$ hold too, and the operators $\cR$ and $\cS$ are history-dependent operators. In addition, assumptions (\ref{E}) and (\ref{u0})  guarantee that conditions $(B)$  and $(u_0)$ are satisfied. It follows from above that we are in a position to apply Corollary \ref{cor3n} to conclude the proof.
\end{proof}

\section{Concluding remarks}
Using tools from convex analysis and fixed points theory, we obtained existence and uniqueness results for a class of time-dependent inclusions in Hilbert spaces. These results were used  to provide the unique solvability of a new class of Moreau's first order sweeping processes with constraints in velocity. Our results are of  interest in the study of quasistatic mathematical models of contact with deformable bodies. Two frictionless and a frictional viscoelastic contact problems were introduced in oder to illustrate these abstract results. Nevertheless, several questions and problems still remain open and need to be investigated in the future. One of these questions is the following: is the smallness condition \eqref{smal} an intrinsic condition in the study of Problem \ref{p1} or it is only a mathematical tool? An open problem is  to extend our results in the case when the data has an $L^p$-regularity, with $p\in [1,+\infty]$. Note that, in this case, there is a need to replace the fixed point Theorem \ref{t00} with an appropriate $L^p$-version. The study of second-order evolutionary sweeping processes would be a valuable extension of the result of this paper. In addition, problems related to the optimal control of time-dependent inclusions and sweeping processes of the form \eqref{i} and \eqref{in}, respectively, represent a topic which deserves to be addressed in the future. All these issues would open the way to important applications in Contact Mechanics.


\section*{Acknowledgement}

\indent This project has received funding from the European Union's Horizon 2020
Research and Innovation Programme under the Marie Sklodowska-Curie
Grant Agreement No 823731 CONMECH.


\begin{thebibliography}{11}
	

\bibitem{AH} S. Adly and T.  Haddad, An implicit sweeping process approach to quasistatic evolution variational inequalities, {\it SIAM J. Math. Anal.} 50 (2018), no. 1, 761-778.

\bibitem{AHT} S. Adly, T. Haddad and L. Thibault, Convex sweeping process in the framework of measure differential inclusions and evolution variational inequalities, {\it Math. Program}. Ser. B 148 (2014), 5--47.

\bibitem{KM} M. Kunze and M. D. P. Monteiro Marques,  On discretization of degenerate osweeping process, {\it Portigalliae Mathematica} 55 (1998), 219--232.





\bibitem{AC} A. Capatina, {\it Variational Inequalities 
	Frictional Contact Problems}, Advances in Mechanics and Mathematics {\bf 31}, Springer, New York, 2014.












\bibitem{DMP1}
Z. Denkowski, S. Mig\'orski and N.S. Papageorgiou,
{\it An Introduction to Non\-li\-near Analysis: Theory}, Kluwer
Academic/Plenum Publishers, Boston, Dordrecht, London, New York,
2003.



\bibitem{DL}
G. Duvaut and J.-L. Lions, {\it Inequalities in Mechanics and
	Physics\/}, Springer-Verlag, Berlin, 1976.

\bibitem{EJK}
	C. Eck, J. Jaru\v sek and M. Krbe\v c, {\it Unilateral Contact
		Problems: Variational Methods and Existence Theorems}, Pure and
	Applied Mathematics {\bf 270}, Chapman/CRC Press, New York, 2005.
	
\bibitem{ET}
	I. Ekeland and R. Temam, {\it Convex Analysis and Variational
		Problems\/}, North-Holland, Amsterdam, 1976.
	
	

\bibitem{HMS}
W. Han, S. Mig\'orski and M. Sofonea, Eds., {\it Advances in Variational and Hemivariational Inequalities},	Advances in Mechanics and Mathematics {\bf 33}, Springer, New York, 2015.
	

\bibitem{HS}
W. Han and M. Sofonea, {\it Quasistatic  Contact Problems in
	Viscoelasticity and Viscoplasticity}, Studies in Advanced
Mathematics {\bf 30}, Americal Mathematical Society, Providence,
RI--International Press, Somerville, MA, 2002.



	

\bibitem{KO} N. Kikuchi and J.T. Oden, {\it Contact Problems in Elasticity: A Study of Va\-ria\-tio\-nal Inequalities and Finite Element Methods}, SIAM, Philadelphia, 1988.

\bibitem{KZ}
A.J. Kurdila and M. Zabarankin, {\it Convex Functional Analysis},
Birk\-h$\ddot{\rm a}$user, Basel, 2005.



\bibitem{MOSBOOK}
S. Mig\'orski, A. Ochal and M. Sofonea,
{\it Nonlinear Inclusions and Hemivariational Inequalities.
	Models and Analysis of Contact Problems},
Advances in Mechanics and Mathematics \textbf{26},
Springer, New York, 2013.






\bibitem{M1} J. J. Moreau, Sur l'\'evolution d'un syst\`eme \'elastoplastique, C. R. Acad. Sci. Paris, S\'er A-Bn 273 (1971), A118--A121.

\bibitem{M2} J. J. Moreau On unilateral constraints, friction and plasticity, in {\it New Variational Techniques in Mathemaical Physics}, G. Capriz and G. Stampacchia, Ed., C.I.M.E. II, Ciclo 1973, Edizione Cremonese, Roma, 1974, p. 173--322.

\bibitem{M3} J. J. Moreau, Intersection of moving convex sets in a normed space, {\it Mat. Scan.} 36 (1975), 159--173.



\bibitem{M4} J. J. Moreau, Evolution problem associated with a moving convex in a Hilbert space, {\it J. Diff. Eqs.} 26 (1977), 347--374.









\bibitem{NP}
Z. Naniewicz and P.D. Panagiotopoulos,
{\it Mathematical Theory of Hemivariational Inequalities
	and Applications}, Marcel Dekker, Inc., New York, Basel, Hong Kong, 1995.
	

\bibitem{Pa1} P.D. Panagiotopoulos,
{\it Inequality Problems in Mechanics and Applications\/},
Birkh\"{a}user, Boston, 1985.


\bibitem{Pa2}
P. D. Panagiotopoulos, {\it Hemivariational Inequalities, 
	Applications in Mechanics and Engineering}, Springer-Verlag,
Berlin, 1993.
	



	


	
\bibitem{SM}
M. Sofonea and A. Matei,
{\it Mathematical Models in Contact Mechanics},
London Mathematical Society Lecture Note Series {\bf 398},
Cambridge University Press, 2012.

\bibitem{SMBOOK}
M. Sofonea and S. Mig\'orski,
{\it Variational-Hemivariational Inequalities with Applications}, Pure and Applied Mathematics, Chapman \& Hall/CRC
Press, Boca Raton-London, 2018.








\end{thebibliography}
\end{document}